\documentclass[a4paper,10pt,twoside,reqno,nonamelimits]{article}
\usepackage{a4wide}
\usepackage[colorlinks=true,urlcolor=blue,linkcolor=blue,citecolor=blue]{hyperref}
\usepackage{newcent}         

\usepackage{amsmath}
\usepackage{amssymb}
\usepackage{amsfonts}
\usepackage{amscd}
\usepackage{amsthm}
\usepackage{amsxtra}
\usepackage{mathrsfs}
\usepackage{nicefrac}
\usepackage{graphicx}
\usepackage{xcolor}
\usepackage{tikz}
\usepackage{enumitem}
\setlist[enumerate]{leftmargin=2em,itemindent=0em, labelindent=0pt,labelwidth=1.5em,labelsep=.5em, align=left, noitemsep}
\newlist{txtenum}{enumerate}{1}
\setlist[txtenum]{leftmargin=0em,itemindent=1.5em, labelindent=0pt,labelwidth=1em,labelsep=.5em, align=left}

\theoremstyle{plain}
\newtheorem{theorem}{Theorem}
\newtheorem*{theorem*}{Theorem}

\newtheorem*{proposition*}{Proposition}

\newtheorem*{corollary*}{Corollary}

\newtheorem{lemma}[theorem]{Lemma}
\newtheorem*{lemma*}{Lemma}

\newtheorem*{observation*}{Observation}

\newtheorem*{conjecture*}{Conjecture}

\newtheorem*{question*}{Question}

\newtheorem*{questions*}{Questions}

\newtheorem*{problem*}{Problem}

\newtheorem*{problems*}{Problems}

\newtheorem*{openproblem*}{Open Problem}

\theoremstyle{definition}
\newtheorem{definition}[theorem]{Definition}
\newtheorem*{definition*}{Definition}

\newtheorem*{example*}{Example}

\newtheorem*{exercise*}{Exercise}

\newtheorem*{remark*}{Remark}

\newtheorem*{remarks*}{Remarks}

\theoremstyle{remark}

\newtheorem*{claim*}{Claim}

\newcommand{\subclass}[1]{}

\newcommand{\enumTi}[1]{\renewcommand{\theenumi}{#1}}

\newcommand{\alphenumi}{\enumTi{\alph{enumi}}}

\newcommand{\romenumi}{\enumTi{\roman{enumi}}}

\newlength{\hspaceforlengthglumpf}

\renewcommand{\em}{\sl}

\newcommand{\lt}{\left}
\newcommand{\rt}{\right}

\newcommand{\abs}[1]{{\lt\lvert{#1}\rt\rvert}}

\newcommand{\RR}{\mathbb{R}}

\newlength{\algotabbingwidth}
\newcommand{\Sets}{\mathscr S}
\newcommand{\Trees}{\mathscr T}
\begin{document}
\title{Fooling Sets and the Spanning Tree Polytope}
\author{Kaveh Khoshkhah, Dirk Oliver Theis\thanks{Supported by the Estonian Research Council, ETAG (\textit{Eesti Teadusagentuur}), through PUT Exploratory Grant \#620, and by the European Regional Development Fund through the Estonian Center of Excellence in Computer Science, EXCS.}\\[1ex]
  \small Institute of Computer Science {\tiny of the } University of Tartu\\
  \small \"Ulikooli 17, 51014 Tartu, Estonia\\
  \small \texttt{\{kaveh.khoshkhah,dotheis\}@ut.ee}%
}

\date{Mon Jan  2 11:58:01 EET 2017}
\maketitle
\begin{abstract}
  In the study of extensions of polytopes of combinatorial optimization problems, a notorious open question is that for the size of the smallest extended formulation of the Minimum Spanning Tree problem on a complete graph with~$n$ nodes.  The best known lower bound is $\Omega(n^2)$, the best known upper bound is $O(n^3)$.

  In this note we show that the venerable fooling set method cannot be used to improve the lower bound: every fooling set for the Spanning Tree polytope has size $O(n^2)$.
  \par\medskip%
  \textbf{Keywords:}  Polyhedral Combinatorial Optimization, Extended Formulations, Lower Bounds.
\end{abstract}

\section{Introduction}\label{sec:intro}
The Spanning Tree polytope $P_n$ has as its vertices the characteristic vectors $\RR^{\binom{[n]}{2}}$ of edge-sets of trees with node set $[n]:=\{1,\dots,n\}$.  A complete system of inequalities and equations is the following:
\begin{subequations}
  \begin{align*}
    \sum_{e\in\binom{[n]}{2}} x_e &=   n-1           &&\\
    \sum_{e\in\binom{S}{2}} x_e   &\le \abs{S}-1     &&\forall S\subset [n],\ \abs{S}>1 \\
    x_e                           &\ge 0             &&\forall e\in\binom{[n]}{2}.
  \end{align*}
\end{subequations}
This system has exponentially many facet-defining inequalities.  There is a classical extended formulation by Wong/Martin~\cite{Martin:91,Wong:tsp:1980} with $O(n^3)$ inequalities (and variables).  A notorious open problem in polyhedral combinatorial optimization, highlighted by M.~Goemans at the 2010 \textit{Carg\`ese Workshop on Combinatorial Optimization,} asks whether or not an extended formulation with $o(n^3)$ inequalities exists.  The only known lower bound is $\Omega(n^2)$ --- this is a ``trivial'' lower bound (it is the dimension of the spanning tree polytope).

One method for proving lower bounds to sizes of extended formulations is the so-called fooling-set method (see, e.g., \cite{Kushilevitz-Nisan:Book:97}).  For that method, one looks at the function~$f$ which maps every pair of an inequality and a solution to~$0$, if the solution satisfies the inequality with equation, and otherwise to~$1$.  A \textit{fooling set of size~$r$} is a list of pairs $F = \bigl( (s_1,t_1),\dots,(s_r,t_r) \bigr)$ such that $f(s_j,t_j)=1$ for $j=1,\dots,r$, and $f(s_i,t_j)f(s_j,t_i)=0$ for $i\ne j$.  If a fooling set of size~$r$ exists, then every extended formulation must contain at least~$r$ inequalities.

It is an easy folklore fact that every polytope has a fooling set of whose size the dimension of the polytope.  But fooling set lower bounds have their limitations.  Most importantly, there can never be a fooling set larger than the square of the dimension of the polytope (\cite{Fiorini-Kaibel-Pashkovich-Theis:CombLB:13}; see also \cite{Dietzfelbinger-Hromkovic-Schnitger:96,Pourmoradnasseri-Theis:foolrk:arXiv}).  However, since that bound is known to be tight \cite{FriesenTheis13,Friesen-Hamed-Lee-Theis:fool:15}, and the dimension of the spanning tree polytope is $\Omega(n^2)$, a fooling set bound of $\Omega(n^3)$ for the Spanning Tree polytope is possible.  A small improvement was made in~\cite{Stefan:phd}, where a $O(n^{8/3}\log n)$ upper bound for the largest possible fooling set for the Spanning Tree polytope is shown.

In this note, we prove that the largest size of a fooling set in the Spanning Tree polytope is $O(n^2)$.

\begin{theorem}\label{thm:main}
  Fooling sets for the spanning tree polytope have size $O(n^2)$.
\end{theorem}

We believe that our tools may be useful for other attempts to give upper bounds on sizes of fooling sets.

\section{Proof of Theorem~\ref{thm:main}}
We start by making precise the function~$f$.  Since the number of nonnegativity inequalities ``$x_e\ge 0$'' is $O(n^2)$, and we aim for an $O(n^2)$ upper bound, we can omit them.  We define
\begin{align*}
  \Sets  &:= \Bigl\{ S \subsetneq [n]             \Bigm| \abs{S}>1       \Bigr\} \\
  \Trees &:= \Bigl\{ T \subset    \binom{[n]}{2}  \Bigm| T \text{ tree } \Bigr\}.
\end{align*}
The function~$f$ is now the following:
\begin{equation*}
  f\colon \Sets\times \Trees \colon
  (S,T) \mapsto
  \begin{cases}
    0, &
      \begin{aligned}[t]
        &\text{if the sub-forest of~$T$ induced by~$S$ is a tree,}\\
        &\text{i.e.,  $T\cap \textstyle\binom{S}{2}$ is connected;}\\[.5ex]
      \end{aligned}\\[.5ex]
    1, & \text{ otherwise.}
  \end{cases}
\end{equation*}
We use the terminology ``$S$ is connected in~$T$'' for the case $f(S,T)=0$.

Now let $F = \bigl( (S_1,T_1)\dots,(S_r,T_r) \bigr)$ be a fooling set of size~$r$ for~$f$.  We define an $(r\times r)$-matrix $H$ by $H_{i,j} := f(S_i,T_j)$.
We will prove the following.

\begin{lemma}[Main Lemma]\label{lem:main}
  Every column of~$H$ has at most $n(n-1)$ zero entries.
\end{lemma}

The Main Lemma immediately implies Theorem~\ref{thm:main}.

\begin{proof}[Proof of Theorem~\ref{thm:main} from Lemma~\ref{lem:main}]
  That~$F$ is a fooling set means that $H$ (has $1$s on the diagonal, and) in each pair of diagonally opposite off-diagonal entries, at least one is zero.  So $H$ has at least $\Omega(r^2)$ zero entries.  But the Main Lemma implies that~$H$ has at most $O(rn^2)$ zero entries.
\end{proof}

\subsubsection*{In the remainder of the section, we will prove the Main Lemma.}
We start with the fundamental definition of the proof.

\begin{definition}[Witness]\label{def:witness}
  Let $S\in\Sets$, $T\in\Trees$.  A \textit{witness for $f(S,T)=1$} is a triple $(a,x,b)$ of nodes which satisfies
  \begin{enumerate}[label=(\alph*)]
  \item\label{def:witness:S} $a,b\in S$, $x\notin S$; and
  \item\label{def:witness:T} $x$ lies on the (unique) path in~$T$ between $a$ and~$b$.
  \end{enumerate}
\end{definition}
Let us justify the terminology.  Clearly, if a witness exists, $S$ cannot be connected in~$T$, and hence $f(S,T)=1$.  On the other hand, if $f(S,T)=1$, i.e., the sub-forest of~$T$ induced by~$S$ has multiple connected components, then there is a node $x\notin S$ with the property that for every $a\in S$ there is a~$b\in S$ such that~$x$ lies on the path in~$T$ between $a$ and~$b$. (To find such an~$x$ shrink the connected components to ``super nodes'', and let~$x$ be any (non-shrunk) node on a path between two ``super nodes''.)  Hence, a witness for $f(S,T)=1$ exists if and only if $f(S,T)=1$.

Clearly, if $(a,x,b)$ is an witness for $f(S,T)=1$, then so is $(b,x,a)$.

\begin{lemma}\label{lem:dj}
  If $(a,x,b)$ is a witness for both $f(S_i, T_i)=1$ and $f(S_j, T_j)=1$ then $i=j$.
\end{lemma}
\begin{proof}
  If $(a,x,b)$ is a witness for both $f(S_i, T_i)=1$ and $f(S_j, T_j)=1$, then it is also a witness for both $f(S_i,T_j)=1$ and $f(S_j,T_i)=1$.  But since~$F$ is a fooling set, this can only happen if $i=j$.
\end{proof}

\begin{lemma}\label{lem:triangle}
  Let $S\in\Sets$, $T\in\Trees$, and $a,b,c\in S$.  If at least one of $(a,x,b)$, $(b,x,c)$, $(c,x,a)$ is a witness for $f(S,T)=1$, then at least two of them are.
\end{lemma}
\begin{proof}
  Condition~\ref{def:witness:S} of Definition~\ref{def:witness} is satisfied by either none or all of the triples.  As for Condition~\ref{def:witness:T}, if say, $x$ were on the path between $a$ and~$b$ in~$T$, but~$x$ was neither on the path between $b$ and~$c$ nor on the path between $c$ and~$a$, then~$T$ would have two distinct paths between $a$ and~$b$, one through~$x$ and one avoiding~$x$ --- a contradiction.
\end{proof}

\begin{lemma}\label{lem:tree-witness}
  Let $S\in\Sets$, $T,T'\in\Trees$, and $x\in[n]$ such that $W_{x,S,T}\ne \emptyset$ and~$S$ is connected in~$T'$.  Then there is an edge $\{a,b\}\in T'\cap\binom{S}{2}$ with $(a,x,b)\in W_{x,S,T}$.
\end{lemma}
\begin{proof}
  Consider the tree $Q$ with node set~$S$ and edge set $T'\cap\binom{S}{2}$.  Among all pairs $\{a,b\}$ with $(a,x,b)\in W_{x,S,T}$ take one for which the distance between $a$ and~$b$ in $Q$ is minimal.  We claim that $\{a,b\}\in Q$.  Indeed, by Lemma~\ref{lem:triangle}, if there was a~$c$ on the path in~$Q$ between $a$ and~$b$ with $c\notin\{a,b\}$, then for at least one of $\{a,c\}$ or $\{c,b\}$ the corresponding triplet with~$x$ in the middle is in $W_{x,S,T}$, and the distance in~$Q$ is shorter.
\end{proof}

For every $x\in[n]$ (a node) and $k\in[r]$ (a column of~$H$), denote by $Z_{x,k}$ the set of $i\in[r]$ (rows of~$H$) for which
\begin{itemize}
\item $S_i$ is connected in $T_k$; and
\item there exist $a,b \in S_i$ such that $(a, x, b)$ is a witness for $f(S_i,T_i)=1$.
\end{itemize}

\begin{lemma}\label{lem:key}
  For every $x\in[n]$ and $k\in[r]$, we have $\abs{Z_{x,k}} \le n-1$.
\end{lemma}
\begin{proof}
  For each $i\in Z_{x,k}$, by Lemma~\ref{lem:tree-witness}, there is an edge $\{a_i,b_i\}\in T_k$ such that $(a,x,b)\in W_{x,S_i,T_i}$.

  By Lemma~\ref{lem:dj}, the edges $\{a_i,b_i\}$, $i\in Z_{x,k}$, are distinct.  Hence $\abs{ Z_{x,k} } \le n-1$.
\end{proof}

The proof of the Main Lemma is now almost complete.

\begin{proof}[Proof of Lemma~\ref{lem:main}]
  Let $k\in[r]$ be the index of a column of~$H$.  For every $i\in [r]$ with $f(S_i,T_k)=0$, there is an~$x$ such that $i \in Z_{x,k}$.  Hence the number of zero-entries in that column of~$H$ is at most
  \begin{equation*}
    \sum_{x=1}^n \abs{ Z_{x,k} },
  \end{equation*}
  which, by Lemma~\ref{lem:key} is at most $n(n-1)$.
\end{proof}

\section{Conclusion and Outlook}
We have settled the question whether the fooling set method can be used to prove a non-trivial lower bound for the extension complexity (i.e., smallest number of inequality in an extended formulation) of the Spanning Tree polytope.  There are two more, stronger, combinatorial lower bounds which are frequently used: the rectangle covering number and the fractional rectangle covering number.  The next step would be to find out if any of these grow more quickly than $n^2$.  In summary, though, despite the interest in the problem (e.g., \cite{Beasley-Klauck-Lee-Theis:Dagstuhl:13,Klauck-Lee-Theis-Thomas:Dagstuhl:15}) and partial results and approaches (e.g., \cite{Faenza-Fiorini-Grappe-Tiwary:nngrk-rCC:2012,Fiorini-Huynh-Joret-Pashkovich:spanningtree:2017}), the problem is now more open than ever.

\subsection*{Acknowledgments}
This research was supported by the Estonian Research Council, ETAG (\textit{Eesti Teadusagentuur}), through PUT Exploratory Grant \#620.  We also gratefully acknowledge funding by the European Regional Development Fund through the Estonian Center of Excellence in Computer Science, EXCS.


\end{document}